\documentclass[11pt,reqno]{amsart}
\usepackage{latexsym,amsmath}

\usepackage{setspace}
\usepackage[a4paper,top=33mm, bottom=27mm, left=20mm, right=20mm]{geometry}

\usepackage{amsthm}
\usepackage{amssymb}
\usepackage{amsthm}
\usepackage{epsfig}
\usepackage{mathrsfs}
\usepackage[usenames,dvipsnames]{xcolor}
\usepackage[english]{babel}

\newtheorem{theorem}{Theorem}[section]
\newtheorem{corollary}[theorem]{Corollary}
\newtheorem{lemma}[theorem]{Lemma}

\newtheorem{proposition}[theorem]{Proposition}
\theoremstyle{definition}

\newtheorem{remark}[theorem]{Remark}

\newcommand{\E}{\mathbb E}

\newcommand{\remove}[1]{}

\DeclareMathOperator{\diam}{diam}
\newcommand{\ent}{\ensuremath {\mathbb Z} }

\newcommand{\dg}{d_{G}}

\newcommand{\de}{d_{E}}

\newcommand{\RG} {\ensuremath{\mathscr G(n,r)}}
\newcommand{\RGuv} {\ensuremath{\widetilde{\mathscr G}_{u,v}(n,r)}}
\newcommand{\RGR} {\ensuremath{\widetilde{\mathscr G}_{R,u,v}(n,r)}}
\newcommand{\SR} {\ensuremath{\mathcal S_n}}

\def\wck#1 {\underline{#1}~\marginpar{\fbox{#1} {\tiny ?}}}
\def\silent#1\par{\par}
\def\text#1{\quad\mbox{#1}\quad}

\makeatletter
\renewcommand{\@seccntformat}[1]{\@nameuse{the#1}.\quad}
\makeatother

\begin{document}

\title{On the relation between graph distance and Euclidean distance in random geometric graphs}
\thanks{Partially
supported by the  CYCIT: TIN2007-66523 (FORMALISM).}

\author{J.~D\'{i}az}
\address{Universitat Polit\`{e}cnica de Catalunya, Dept. de Llenguatges i Sistemes Inform\`{a}tics,
08034 Barcelona}
\email{\texttt{diaz@lsi.upc.edu}}

\author{D.~Mitsche}
\address{Universit\'{e} de Nice Sophia-Antipolis, Laboratoire J.A. Dieudonn\'{e}, Parc Valrose,
06108 Nice cedex 02}
\email{\texttt{dmitsche@unice.fr}}

\author{G.~Perarnau}
\address{Universitat Polit\`{e}cnica de Catalunya, Dept. de Matem\`{a}tica Aplicada IV, 08034
Barcelona}
\email{\texttt{guillem.perarnau@ma4.upc.edu}}

\author{X.~P\'{e}rez-Gim\'{e}nez}
\address{University of Waterloo, Dept. of Combinatorics and Optimization, Waterloo ON, N2L 3G1}
\email{\texttt{xperez@uwaterloo.ca}}

\keywords {random graphs}
\subjclass {Primary: 05C80}
\date{}

\maketitle
\onehalfspace
\begin{abstract}
Given any two vertices $u,v$ of a random geometric graph, denote by
$d_E(u,v)$ their Euclidean distance and by $d_G(u,v)$ their graph
distance. The problem of finding upper bounds on $d_G(u,v)$ in terms
of $d_E(u,v)$ has received a lot of attention in the
literature~\cite{Bradonjic,Ellis,Friedrich,Muthukrishnan}. In this
paper, we improve these upper bounds for values of
$r=\omega(\sqrt{\log{n}})$ (i.e.~for $r$ above the connectivity threshold). Our result also improves the best known
estimates on the diameter of random geometric graphs. We also
provide a lower bound on $d_G(u,v)$ in terms of $d_E(u,v)$.
\end{abstract}
{\small\textbf{Keywords:} Random geometric graphs, Graph distance, Euclidean distance, Diameter.}

\section{Introduction}\label{ap:sec:intro}
Given a positive integer $n$, and a non-negative real $r$, we consider a random geometric graph
$G\in\RG$ defined as follows.
The vertex set $V$ of $G$ is obtained by choosing $n$ points independently and uniformly at random
in the square $\SR = \left[-\sqrt{n}/2,\sqrt{n}/2\right]^2$ (Note that, with probability $1$, no point in
$\SR$ is chosen more than once, and thus we assume $|V|=n$).
For notational purposes, we identify
each vertex $v \in V$ with its corresponding geometric position $v=(v_x,v_y)\in\SR$, where $v_x$ and $v_y$
denote the usual $x$- and $y$-coordinates in $\SR$. Finally, the edge set of $G\in\RG$ is constructed by
connecting each pair of vertices $u$ and $v$ by an edge if and only if $d_E(u,v)\le r$, where $d_E$
denotes the Euclidean distance in $\SR$.

Random geometric graphs were first introduced in a slightly different setting by
Gilbert~\cite{Gilbert} to model the communications between radio stations. Since then, several
closely related variants on these graphs have been widely used as a model for wireless
communication, and have also been extensively  studied from a mathematical point of view. The basic
reference on random geometric graphs is the monograph by Penrose~\cite{Penrose}.

The properties of $\RG$ are usually investigated from an asymptotic perspective, as $n$ grows to
infinity and $r=r(n)$. Throughout the paper, we use the following standard notation for
the asymptotic behavior of sequences of non-negative numbers $a_n$ and $b_n$: $a_n=O(b_n)$ if
$\limsup_{n\to\infty}a_n/b_n\leq C<+\infty$;  $a_n=\Omega(b_n)$ if $b_n=O(a_n)$;
$a_n=\Theta(b_n)$ if $a_n=O(b_n)$ and $a_n=\Omega(b_n)$; $a_n=o(b_n)$ if
$\lim_{n\to\infty} a_n/b_n = 0$. Finally,  a sequence of events $H_n$ holds \emph{asymptotically
almost surely} (a.a.s.) if $\lim_{n\to\infty}\Pr(H_n)=1$.

It is well known that  $r_c=\sqrt{\log n/\pi}$ is a sharp threshold function for the connectivity of
a random geometric graph (see~e.g.~\cite{Penrose97,Goel05}). This means that for every $\varepsilon>0$, if
$r\le(1-\varepsilon)r_c$, then $\RG$ is a.a.s.\ disconnected, whilst if
$r\ge(1+\varepsilon)r_c$, then it is a.a.s.\ connected.
In order to ensure that we have a connected random geometric graph, we assume in the following that
$r \geq r_c$.

Given a connected graph $G$, we define the {\em graph distance} between two vertices $u$ and $v$, denoted by
$d_G(u,v)$, as the number of edges on a shortest path from $u$ to $v$. Observe first that any pair of vertices $u$ and $v$ must satisfy $d_G(u,v)\geq d_E(u,v)/r$ deterministically, since each edge of a geometric graph has length at most $r$. The goal of this paper is to provide upper and lower bounds that hold a.a.s.\ for the graph distance of two vertices in terms of their Euclidean distance and in terms of $r$ (see Figure~\ref{ap:fig:distxy2}).
\vspace{0.3cm}

\emph{Related work.} The particular problem has risen quite a bit of interest in recent years. Given any two $v,u\in V$, most of the work related to this problem has been devoted to study upper bounds on $d_G(u,v)$ in terms of $d_E(u,v)$ and $r$, that hold a.a.s. Ellis, Martin and Yan~\cite{Ellis} showed that there exists some large constant $K$ such that for every $r\geq r_c$, $G\in \RG$ satisfies a.a.s.\ $d_G(u,v)\leq K d_E(u,v)/r$ for every $u$ and $v$\footnote{The result is stated in the unit ball random geometric graph model, but can be adapted to our setting.}. This result was extended by Bradonjic et al.~\cite{Bradonjic} for the range of $r$ for which  $\RG$ has a giant component a.a.s., under the extra condition that $d_E(u,v)=\Omega(\log^{7/2}n/r^2)$. Friedrich, Sauerwald and Stauffer~\cite{Friedrich} improved this last result by showing that the result holds a.a.s.\ for every $u$ and $v$ satisfying $d_E(u,v)=\omega(\log{n}/r)$. They also proved that if $r=o(r_c)$,  a linear upper bound of $d_G(u,v)$ in terms of $d_E(u,v)$ is no longer possible. In particular, a.a.s.\ there exist vertices $u$ and $v$ with $d_E(u,v)\leq 3r$ and $d_G(u,v)=\Omega(\log{n}/r^2)$.

The motivation for the study of this problem stems from the fact that these results provide upper bounds for the diameter of $G\in \RG$, denoted by $\diam(G)$, that hold a.a.s., and the runtime complexity of many algorithms can often be bounded from above in terms of the diameter of $G$. For a concrete example, we refer to the problem of broadcasting information (see~\cite{Bradonjic,Friedrich}).

One of the important achievements of our paper is to show that $K=1+o(1)$ a.a.s., provided that $r=\omega(r_c)$. By the result in~\cite{Friedrich}, we know that such a result is false if $r=o(r_c)$.

A similar problem has been studied by Muthukrishnan and Pandurangan~\cite{Muthukrishnan}. They proposed a new technique to study several problems on random geometric graphs --- the so called  \emph{Bin-Covering} technique --- which tries to cover the endpoints of a path by bins. They consider, among others, the problem of determining $D_G(u,v)$, which is the length of the shortest Euclidean path connecting $u$ and $v$.
Recently, Mehrabian and Wormald~\cite{Abbas} studied a similar problem to the one in~\cite{Muthukrishnan}. They deploy $n$ points uniformly in $[0,1]^2$, and connect any pair of points with probability $p=p(n)$, independently of their distance. Mehrabian et al. determine the ratio of $D_G(u,v)$ and $d_E(u,v)$ as a function of $p$.

%


\begin{figure}[t]
\begin{center}
 \scalebox{.6}{\includegraphics{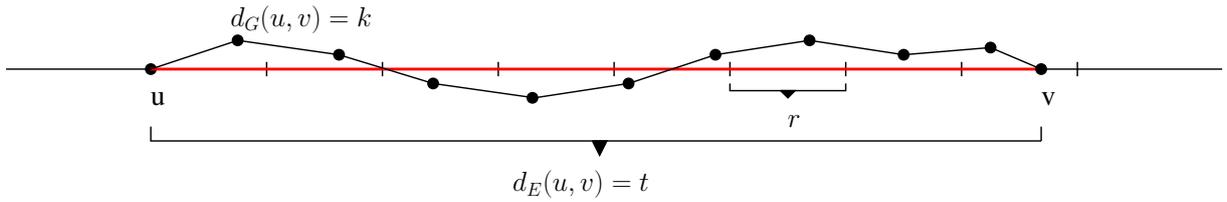}}
\end{center}
\caption{Graph distance vs. Euclidean distance between two points $u$ and $v$ in
$V$}\label{ap:fig:distxy2}
\end{figure}

The following theorem is the main result of the paper.
\begin{theorem}\label{ap:thm:main}
Let $G\in\RG$ be a random geometric graph with $r \geq r_c$. A.a.s., for every pair of vertices $u,v\in V(G)$ we have:
\begin{enumerate}
\item if $d_E(u,v) \geq 20 r \log n$, then $\displaystyle d_G(u,v) \geq \frac{d_E(u,v)}{r}\left(1 + \frac{1}{2\left(r d_E(u,v)\right)^{2/3}}\right)$, and
\item if $r\geq 70\sqrt{ \log{n}}$, then $\displaystyle d_G(u,v) \leq \left\lceil \frac{d_E(u,v)}{r}\left(1+\gamma r^{-4/3}\right) \right\rceil$
\end{enumerate}
where
$$\gamma=\max\left\{31\left(\frac{2r\log{n}}{r+d_E(u,v)}\right)^{2/3},\;\frac{70\log^2{n} }{r^{8/3}},\; 300^{2/3}\right\}\;.$$
%
\end{theorem}
In order to prove (i), we first observe that all the short paths between two points must lie in a certain rectangle. Then we show that, by restricting the construction of the path on that rectangle, no very short path exists.
For the proof of (ii) we proceed similarly. We restrict our problem to finding a path contained in a narrow strip. In this case, we show that a relatively short path can be constructed.
We believe that the ideas in the proof can be easily extended to show the analogous result for $d$-dimensional random geometric graphs for all fixed $d\geq 2$.

\begin{remark}
(1)
We do not know if the condition $d_E(u,v) \geq 20 r \log n$ in the lower bound in (i) can be improved.
(2)
The constant $70$ in the condition $r\ge70\sqrt{\log n}$ of (ii) (as well as those in the definition of $\gamma$) is not optimized, and could be made slightly smaller. However, our method as it is, cannot be extended all the way down to $r\ge\sqrt{\log n/\pi}=r_c$.
(3)
The error term in~(ii) is
\[
\gamma r^{-4/3} = \Theta \left( \max \left\{ \left(\frac{\log{n}}{r^2+rd_E(u,v)}\right)^{2/3},\; \left(\frac{\sqrt{\log{n}} }{r}\right)^4,\; r^{-4/3}\right\} \right),
\]
which is $o(1)$ iff $r=\omega(\sqrt{\log n})=\omega(r_c)$.
Hence, for $r=\omega(r_c)$, statement (ii) implies that a.a.s.
$$
d_G(u,v) \leq \left\lceil \left(1+o(1)\right) \frac{d_E(u,v)}{r}\right\rceil\;,
$$
thus improving the result in~\cite{Ellis}.
\end{remark}

\vspace{0.3cm}

Theorem~\ref{ap:thm:main} gives an upper bound on the diameter as a corollary.
First, observe that $d_E(u,v)\leq \sqrt{2n}$. From Theorem~10 in~\cite{Ellis} for the particular case $d=2$, one can easily deduce that if $r\geq r_c$ a.a.s.
\begin{equation}\label{ap:eq:diam1}
\diam(G)\leq \frac{\sqrt{2n}}{r}\left(1+O\left(\sqrt{\frac{\log\log{n}}{\log{n}}}\right)\right)\;.
\end{equation}
Moreover, note that a.a.s.\ there exist $u$ and $v$ at distance $d_E(u,v)\geq \sqrt{2n}-2\sqrt{2}\log{n}$; the probability that the squares of side $\log{n}$ at the corners of $\SR$ contain no vertices is $o(1)$. Applying Theorem~\ref{ap:thm:main} to these vertices $u$ and $v$, we obtain the following result.
\begin{corollary}\label{ap:cor:diam}
Let $G\in\RG$ be a random geometric graph  with $r\geq 70\sqrt{\log{n}}$. A.a.s.\ we have
$$
\diam(G)\leq \frac{\sqrt{2n}}{r}\left(1+\gamma r^{-4/3}\right)\;,
$$
where $\gamma$ is defined as in Theorem~\ref{ap:thm:main}.
\end{corollary}
Observe that our bound on the diameter stated in Corollary~\ref{ap:cor:diam} improves the one in~\eqref{ap:eq:diam1} derived from~\cite{Ellis} whenever
$r=\Omega\Big(\dfrac{\log^{5/8}n}{(\log\log n)^{1/8}}\Big)$.




%
\section{Proof of Theorem~\ref{ap:thm:main}}\label{ap:sec:main}
In order to simplify the proof of Theorem~\ref{ap:thm:main} we will make use of a technique known as de-Poissonization, which has many applications in geometric probability (see~\cite{Penrose} for a detailed account of the subject). Here we sketch it.

Consider the following related model of a random geometric graph given vertices $u$ and $v$. Let $V=\{u,v\}\cup V'$, where $V'$ is a set obtained as a homogeneous Poisson point process of intensity $1$ in the square $\SR$ of area $n$. In other words, $V'$ consists of $N$ points in the square $\mathcal S_n$ chosen independently and uniformly at random, where $N$ is a Poisson random variable of mean $n$. We add two labelled vertices $u$ and $v$, whose position is also selected independently and uniformly at random in $\SR$. Exactly as we did for the model $\RG$, we connect by an edge $u$ and $v$ in $V$ if $d_G(u,v)\leq r$. We denote this new model by $\RGuv$.

The main advantage of defining $V'=V\setminus\{u,v\}$ as a Poisson point process is motivated by the following two properties:
the number of points of $V'$ that lie in any region $A\subseteq\SR$ of area $a$ has a Poisson distribution with mean $a$; and the number of points of $V'$ in disjoint regions of $\SR$ are independently distributed.
Moreover,  by conditioning $\RGuv$ upon the event $N=n-2$, we recover the original distribution of $\RG$.
Therefore, since $\Pr(N=n-2)=\Theta(1/\sqrt n)$,
any event holding in $\RGuv$ with probability at least $1-o(f_n)$ must hold in $\RG$ with probability at least $1-o(f_n \sqrt n)$.
We make use of this property throughout the article, and do all the analysis for a graph  $G\in \RGuv$.
%



We will need the following concentration inequality for the sum of independently and identically distributed exponential random variables. For the sake of completeness we provide the proof here.
\begin{lemma}\label{ap:lem:concentration}
 Let $X_1,\dots,X_N$ be independent exponential random variables and let
$X=X_1+\dots+X_N$. Then, for every $\delta > 0$ we have
$$
\Pr(X\ge (1+\delta)\E(X)) \le \left( \frac{1+\delta}{e^{\delta}} \right)^{N},
$$
and for any $0<\delta<1$ we have
$$
\Pr(X\le (1-\delta)\E(X)) \le \left( (1-\delta)e^{\delta} \right)^{N}.
$$
\end{lemma}

%
\begin{proof}
 By Markov's inequality, we have for every $\beta > 0$
$$
\Pr(X\geq(1+\delta)\E(X)) = \Pr(e^{\beta X}\ge e^{\beta(1+\delta)\E(X)}) \le
\frac{\prod \E(e^{\beta X_i})}{e^{\beta(1+\delta)\E(X)}}=(\varphi_{X_1}(\beta))^N
e^{-\beta(1+\delta)N/\mu}\;,
$$
where $\varphi_{X_1}(\beta)=\E(e^{\beta X_1})=\frac{\mu}{\mu-\beta}$ is the moment-generating
function of an
exponentially distributed random variable with parameter $\mu$.
Thus,
\begin{align*}
\Pr(X\geq (1+\delta)\E(X)) &\le \left(\frac{\mu}{\mu-\beta}\right)^N e^{-\beta(1+\delta)N/\mu}\\
&= \exp\left(N\left(-\log\left(1-\frac{\beta}{\mu}\right) - (1+\delta)\frac{\beta}{\mu} \right)
\right)\;.
\end{align*}
Setting $\frac{\beta}{\mu} =\frac{\delta}{1+\delta}$, we have
\begin{align*}
\Pr(X\geq(1+\delta)\E(X)) &\le \exp\left(N \left(\log\left(1+\delta\right) -
\delta\right)\right) = \left( \frac{1+\delta}{e^{\delta}} \right)^{N}.
\end{align*}
The lower tail is proved similarly.
%
\end{proof}
\subsection{Proof of statement (i)}\label{ap:ssec:i}
Our argument in this subsection depends only on the Euclidean distance between $u$ and $v$, but not on their particular
position in $\SR$. Thus, let $t=d_E(u,v)$ and assume without loss of generality that $u=(0,0)$ and
$v=(t,0)$.

The next lemma shows that short paths between vertices are contained in small strips. It is stated in the more general context of a
geometric graph $G=(V,E)$ of radius $r$, where the vertex set $V$ is a subset of points in the
square $\SR$ (not necessarily randomly placed), and edges connect (as usual) every pair of vertices at
Euclidean distance at most $r$. For every $\alpha>0$, consider the rectangle $R = [0,t]\times [-\alpha,\alpha].$
%


\begin{lemma}\label{ap:lem:strip}
Let $G=(V,E)$ be a geometric graph with radius $r$ in $\SR$, and let $u,v\in V$ such that $u=(0,0)$
and $v=(t,0)$. Suppose that $t=d_E(u,v)  \ge  kr-\frac{2\alpha^2}{kr}$, for some $k\in\ent^+$ and $\alpha=o(kr)$. Then all paths of length at
most $k$ from $u$ to $v$ are contained in $R$.
\end{lemma}
\begin{proof}
Suppose that there exists a path from $u$ to $v$ in at most $k$ steps. Let $z=(a,b)$ the vertex with largest $y$-coordinate in
that path.
Since $a\in [0,t]$, for any $b$ we have,
$$
kr\geq \sqrt{a^2+b^2}+\sqrt{(t-a)^2+b^2} \geq  2\sqrt{t^2/4+b^2} \;.
$$
Therefore,
$$
\frac{(kr)^2}{4} \geq \frac{t^2}{4}+b^2 \ge \frac{(kr-\frac{2\alpha^2}{kr})^2}{4}+b^2\;,
$$
where we used that $t \ge kr-\frac{2\alpha^2}{kr}$. Using that $\alpha = o(kr)$ we have
$$
b \leq  \alpha \sqrt{1-\alpha^2/(kr)^2}=(1-o(1))\alpha\;.
$$
Repeating the same argument for the vertex with smallest $y$-coordinate, we conclude that the path is contained in $R=[0,t]\times
[-\alpha,\alpha]$.
\end{proof}

\begin{proposition}\label{ap:prop:lb}

Let $G\in\RGuv$ be a random geometric graph on $\SR$, with $u=(0,0)$ and $v=(t,0)$.
Then, for every $0<\delta<2^{-1/3}$, we have that
\begin{align}\label{ap:eq:cond}
\Pr \left( d_G(u,v)  \leq \frac{t}{r}\left(1 +\frac{\delta}{(tr)^{2/3}} \right) \right) \le \frac{t}{r} \exp\left(- \sqrt{\delta/2}(tr)^{2/3}\right)+
\exp\left(-(1-\sqrt{2\delta^3})^2 \frac{t}{2r}\right).
\end{align}
\end{proposition}

\begin{proof}
Let $k=d_G(u,v)$ and let $\alpha=\sqrt{\delta/2}\left((d_G(u,v))^3r^2/t\right)^{1/3}=o(d_G(u,v) r)$.
Consider the event $A_{\alpha}$ that all the paths from $u$ to $v$ of length $k$ are contained in the rectangle
$R=[0,t]\times [-\alpha,\alpha]$ and let $B$ the event defined by condition~\eqref{ap:eq:cond}.
If $B$ holds, then
\begin{align*}
t&\geq \frac{kr}{\left( 1+\frac{\delta}{(tr)^{2/3}}\right)}
\geq kr\left( 1-\frac{\delta}{(tr)^{2/3}}\right)
= kr-\frac{2\alpha^2}{kr}\;.
\end{align*}

Since $\alpha=o(kr)$, by Lemma~\ref{ap:lem:strip},  $\Pr(B|\overline{A_\alpha})=0$.


Denote by $v_1$ the vertex with largest $x$-coordinate inside the rectangle $R_1=[0,r]\times [-\alpha,\alpha]$
(possibly $v_1=u$ if $R_1$ contains no other vertices of $\RGuv$).
Note that $v_1$ might not be connected to $u$, but observe that its $x$-coordinate is always greater
or equal to the
$x$-coordinate of any vertex $u_1\in R_1$ connected to $u$ (see Figure~\ref{ap:fig:insideR}). Let $x_1$ be the
$x$-coordinate of $v_1$, and define the random variable $a_1=r-x_1$. By definition, $0\le a_1\le r$. Since $G\in \RGuv$, the number of vertices from $V$ inside a region of $\SR$ is a Poisson random variable with mean equal to the area of that region. Hence, the random variable $a_1$ satisfies
\begin{equation}\label{ap:eq:1}
\Pr(a_1\geq \beta) =
\begin{cases}
e^{-2\alpha \beta} & \text{if $0\le\beta\le r$}\\
0 & \text{if $\beta>r.$}
\end{cases}
\end{equation}
Thus, $a_1$ is stochastically dominated by an exponentially distributed random variable $\tilde a_1$
of parameter $2\alpha$. We assume that $a_1$ and $\tilde a_1$ are coupled together in the same
probability space, so that $a_1=\min\{\tilde a_1, r\}\le\tilde a_1$.

We  proceed to define in a similar way the points $v_i$ and the values $x_i$ and $a_i$, for any
$2\leq i\leq k$.
Let $v_i$ be the  vertex with largest $x$-coordinate inside the rectangle  $R_i=(x_{i-1}+a_{i-1}, x_{i-1}+r]\times
[-\alpha,\alpha]$, and let $x_i$ be the $x$-coordinate of $v_i$. Define $a_i=x_{i-1}+r-x_i$. If
$R_i$ contains no vertex of $\RGuv$, then add an \emph{extra} vertex $v_i=(x_{i-1}+a_{i-1},0)$ (so
in that case $x_i=x_{i-1}+a_{i-1}$ and $a_i= r- a_{i-1}$).
Observe that $0\le a_i\le r-a_{i-1}$, so the rectangles $R_1,R_2,\ldots,R_k$ are disjoint. Moreover,
\begin{equation}\label{ap:eq:i}
\Pr(a_i\geq \beta) =
\begin{cases}
e^{-2\alpha \beta} & \text{if $0\le\beta\le r-a_{i-1}$}\\
0 & \text{if $\beta>r-a_{i-1},$}
\end{cases}
\end{equation}
for every $1\le i\le k$ (by defining $a_0=0$). Therefore, $a_1,a_2,\ldots,a_k$ are stochastically
dominated by a sequence $\tilde a_1,\tilde a_2,\ldots,\tilde a_k$ of i.i.d.\ exponentially
distributed random variables of parameter $2\alpha$, such that $a_i=\min\{\tilde a_i,
r-a_{i-1}\}\le\tilde a_i$ for all $1\le i\le k$.

\begin{figure}[ht!]
 \begin{center}
 \includegraphics[width=0.9\textwidth]{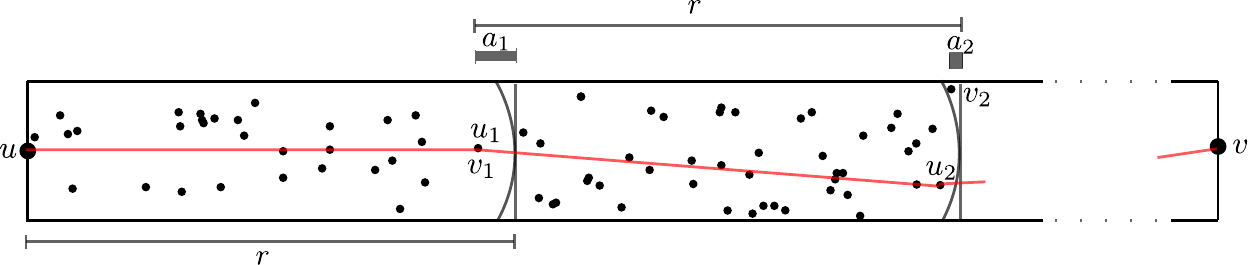}
 \end{center}
 \caption{The points $u_1$, $v_1$ coincide but $u_2$ and $v_2$ do not}\label{ap:fig:insideR}
 \label{ap:fig:path}
\end{figure}

Note that the vertices $u,v_1,v_2,\ldots,v_k$ may not induce a connected path in $\RGuv$, since
the Euclidean distance between two consecutive ones may be greater than $r$.  However, the fact that
$v_i$ is the vertex with largest $x$-coordinate inside  $[0,x_{i-1}+r]\times[-\alpha,\alpha]$ and  together with a straightforward
induction argument yield to the following claim:
if $u=u_0,u_1,u_2\ldots,u_k$ is a path contained in $R$, then for every $1\le i\le k$ the
$x$-coordinate of $u_i$ is at most $x_i$ (see again Figure~\ref{ap:fig:insideR}).
We will now show that $x_k<t$ with the desired probability.

Define
\[
a=\sum_{i=1}^k a_i
\qquad\text{and}\qquad
\tilde a=\sum_{i=1}^k \tilde a_i.
\]
Expanding recursively from the relations $x_i=x_{i-1}+r-a_i$ and $x_1=r-a_1$, we get
\[
x_k = \sum_{i=1}^k (r-a_i) = kr - a.
\]
Let us consider the event that $\tilde a_i\le r/2$ for all $1\le i\le k$. In particular, this event
implies that $a_i=\tilde a_i$ for all $i$, and therefore $x_k = kr - \tilde a$.
Since each $\tilde a_{i}$ is exponentially distributed with parameter $2\alpha$,
\begin{equation}\label{ap:abs:bound}
 \Pr\left(\exists i : a_i > r/2\right) \le k \Pr\left(a_1 > r/2\right) =
ke ^{-\alpha r}\;,
\end{equation}
so $x_k = kr - \tilde a$ with probability at least $1-ke^{-\alpha r }$.

Moreover, notice that
\[
\mathbb{E}(\tilde a)=k\mathbb{E}(\tilde a_1) = \frac{k}{2\alpha}\;.
\]
By Lemma~\ref{ap:lem:concentration}, for all $0<\varepsilon<1$,
\begin{align}\label{ap:eq:previous}
 \Pr\left(\tilde a\le(1-\varepsilon)\frac{k}{2\alpha}\right) &\le
\left((1-\varepsilon)e^{\varepsilon}\right)^{k}
<e^{-\varepsilon^2k/2}\;.
\end{align}
where we have used that $\log(1-x) < -x-\frac{x^2}{2}$, for any  $0<x<1$.

If $t < kr\left(1-\frac{\delta}{(tr)^{2/3}}\right)=kr(1+o(1))$, by definition of $B$, $\Pr(B)=0$ and we are done. Thus, we may assume that $t\geq kr(1+o(1))$.
For $\delta=\left(\frac{1-\varepsilon}{\sqrt{2}}\right)^{2/3}$, from~\eqref{ap:eq:previous} and~\eqref{ap:abs:bound},
\begin{align*}
\Pr\left(x_k \geq t \right)
&\leq \Pr\left(x_k \geq kr\left( 1-\frac{\delta}{(tr)^{2/3}}\right)  \right) \\
&\leq \Pr\left(x_k \geq kr-\frac{\sqrt{2}\delta^{3/2}kr}{t}\cdot\frac{k}{2\alpha}  \right)\\
&\leq \Pr\left(x_k \geq kr-(1-\varepsilon+o(1))\frac{k}{2\alpha}  \right)\\
& \leq \Pr\left(\exists i : a_i > r/2\right) + \Pr\left(\forall i : a_i \leq r/2\text{and}
a\le(1-\varepsilon+o(1))\frac{k}{2\alpha}\right)\\
& \leq \Pr\left(\exists i : a_i > r/2\right) + \Pr\left(\tilde
a\le(1-\varepsilon+o(1))\frac{k}{2\alpha}\right)\\
&\le k e^{-\alpha r }+e^{\varepsilon^2 k/2 } \\
&\le \frac{t}{r} \exp\left(- \sqrt{\delta/2}(tr)^{2/3}\right)+\exp\left(-(1-\sqrt{2\delta^3})^2 \frac{t}{2r}\right)\;.
\end{align*}
Hence, if $A_{\alpha}$ holds, then  $\Pr{(B)} \leq \Pr\left(x_k \geq t \right) \leq \frac{t}{r} \exp\left(- \sqrt{\delta/2}(tr)^{2/3}\right)+\exp\left(-(1-\sqrt{2\delta^3})^2 \frac{t}{2r}\right)$, and if $\overline{A_{\alpha}}$ holds, then $\Pr{(B)}=0$. Thus, the proposition follows.
\end{proof}
%
\begin{proposition}
Let $\RGuv$ be a random geometric graph in $\SR$ with labelled vertices $u$ and $v$ such that
$d_E(u,v) \geq 20 r\log{n}$.
Then we have
$$
d_G(u,v) \le \frac{d_E(u,v)}{r}\left(1+\frac{1}{2(r d_E(u,v))^{2/3}} \right)\;,
$$
with probability at most $o(n^{-5/2})$.
\end{proposition}
\begin{proof}
As before, let $t=d_E(u,v)$ and $k=d_G(u,v)$. Also let $\delta=1/2$.


Since $t\geq 20r\log{n}$ and $r\geq r_c=\Omega(\sqrt{\log{n}})$, we have
$$
\sqrt{\delta/2}(tr)^{2/3} -\log{(t/r)} = \Omega(\log^{4/3}{n})\;,
$$
and
$$
\left(1-\sqrt{2\delta^3}\right)^2 \frac{t}{2r}\geq \frac{5}{2}\log{n}\;.
$$

By Proposition~\ref{ap:prop:lb}, this implies that
$$
\Pr\left(k \leq \frac{t}{r}\left(1 +\frac{1}{2(tr)^{2/3}} \right)\right)=
o(n^{-5/2})\;.
$$
\end{proof}

To finish the proof of statement~(i) in Theorem~\ref{ap:thm:main},
by de-Poissonizing $\RGuv$, we have that in  $\RG$, \remove{technique discussed in the beginning of Section~\ref{ap:sec:main},}
statement~(i) in Theorem~\ref{ap:thm:main} holds for our choice of $u$ and $v$, with probability at
least $1-o(n^{-2})$.
Note that this fact does not depend on the particular location of $u$ and $v$ in $\SR$.
The statement follows by taking a union bound over all at most $n^2$ pairs of vertices.

\subsection{Proof of statement (ii)}\label{ap:ssec:ii}
As in Subsection~\ref{ap:ssec:i}, we pick two points in $\SR$, and put $t=\de(u,v)$. Let $\gamma$ be as
in the statement of Theorem~\ref{ap:thm:main}.
We assume first that $u=(0,0)$ and $v=(t,0)$, and consider a Poisson point process in the rectangle
$R=[0,t]\times[0,\alpha]$, for a certain $\alpha \leq  r$ that will be made precise later.

Let $\RGR$ denote the random geometric graph on the rectangle $R$, to which the points $u$ and $v$
are added.
We will show that the probability of having $\dg(u,v) \ge \frac{\de(u,v)}{r}\left(1+\delta r^{-4/3}
\right)$ decays exponentially in $\delta$.
For each point $z$ in $R$ with $x$-coordinate $s$, define the rectangle
\[
R_z=\left[s,s+\rho\right]\times[0,\alpha], \qquad\text{where}\quad \rho = r-\frac{\alpha^2}{r} .
\]
We need the following auxiliary lemma.
\begin{lemma}\label{ap:lem:rectangle}
For any vertex $z$ in $R$, all vertices in $R_z$ are connected to $z$ (see
Figure~\ref{ap:fig:rectangle}).
\end{lemma}
\begin{proof}
It is enough to show that the upper-left and the bottom-right corner of $R_z$ are at distance at
most $r$. Then all vertices inside $R_z$ are connected to one another, and in particular $z$ is
connected to every vertex in $R_z$.  A sufficient condition for that is
\[
\sqrt{\rho^2 + \alpha^2} \le r,
\]
or equivalently
\[
r\left(1-(\alpha/r)^2\right)=\rho \le \sqrt{r^2 - \alpha^2}=r\sqrt{1-(\alpha/r)^2}.
\]
Since $\sqrt{1-x}>1-x$  for any $0<x<1$, the lemma follows.
\end{proof}

\begin{figure}[ht!]
 \begin{center}
 \includegraphics[width=0.4\textwidth]{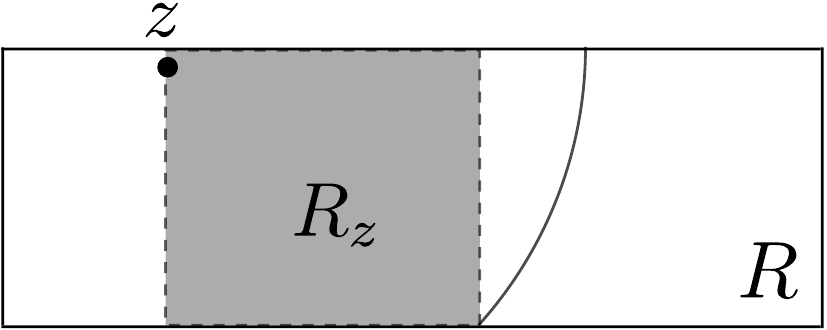}
 \end{center}
 \caption{The rectangle $R_z$}
 \label{ap:fig:rectangle}
\end{figure}

\begin{proposition}\label{ap:alternative}
Let $\RGR$ be a random geometric graph on $R$, with $u=(0,0)$ and $v=(t,0)$.
Let $F>0$ and $J>\frac{3(F+1)}{2^{2/3}}$ be constants and define $g(x)=x-\log(1+x)$.
Then, for every $J\le\delta\leq F r^{4/3}$, we have that
\[
\Pr \left( d_G(u,v) > \left\lceil    \frac{t}{r}\left(1+\delta r^{-4/3} \right)   \right\rceil \right) \le
n \exp\left(-\frac{(F+1) \delta^{1/2} r^{4/3}}{2J^{3/2}}\right) + \exp\left(-g\Big((\delta/J)^{3/2}\Big)\frac{t}{r}\right).
\]
\end{proposition}
\begin{proof}
Set $C=1/J^{3/2}$, and let $B$ be any positive constant satifying
\begin{equation}\label{ap:eq:ineB}
B^2+C/B \le 1/(F+1).
\end{equation}
Some elementary analysis shows that such $B$ must exist. In fact, the equation $B^2+C/B= 1/(F+1)$ has exactly two positive solutions $B_1$ and $B_2$ for any $0<C< \frac{2}{(3(F+1))^{3/2}}$, and any $0<B_1\le B\le B_2<1/\sqrt{F+1}$ satisfies~\eqref{ap:eq:ineB}.

Let us consider the integer $k = \lceil \frac{t}{r}(1+\delta r^{-4/3})\rceil$.
We will show that with very high probability there exists a path length at most $k$ between $u$ and $v$. Such a path will only use vertices inside $R$, but for technical reasons (the last of
the rectangles $R'_i$ defined below might be further to the right than the point $(t,0)$ or possibly
be outside of the square)  of the argument we extend the Poisson point process of our probability
space to the semi-infinite strip $R_{\infty}=[0,\infty)\times[0,\alpha]$.

We construct a sequence of vertices in a similar way as in the proof of Proposition~\ref{ap:prop:lb}.
Set $v_0=u$, $x_0=0$ and $a_0=0$. We make the choice of $\alpha$ for this subsection now more precise. We set
\[
\alpha= B \delta^{1/2}r^{1/3},
\]
for some constant $0<B<1/\sqrt{F+1}$ satisfying~\eqref{ap:eq:ineB}. Observe that the restriction $\delta\leq F r^{4/3}$ implies that
\begin{equation}\label{ap:eq:bound_alpha}
\alpha \le (B \sqrt F) r < r,
\end{equation}
so our choice of $\alpha$ is feasible, and moreover
\begin{equation}\label{ap:eq:bound_rho}
\rho=r-\alpha^2/r \ge (1-B^2F)r.
\end{equation}
For each $1\le i\le k-1$, define $R'_i=(x_{i-1}+a_{i-1},x_{i-1}+\rho]\times[0,\alpha]$, and let
$v_i$ be the vertex with largest $x$-coordinate inside $R'_i$ (if $R'_i$ is empty, then add an \emph{extra} vertex
$v_i=(x_{i-1}+a_{i-1},0)$). Define $x_i$ to be the $x$-coordinate of $v_i$ and
$a_i=x_{i-1}+\rho-x_i$. By the same considerations as in the proof of Proposition~\ref{ap:prop:lb} (but
replacing $R_i$ by $R'_i$, $r$ by $\rho$, and $k$ by $k-1$), we deduce that $a_1,a_2,\ldots,a_{k-1}$
are stochastically dominated by a sequence $\tilde a_1,\tilde a_2,\ldots,\tilde a_{k-1}$ of i.i.d.\
exponentially distributed random variables of parameter $\alpha$, such that $a_i=\min\{\tilde a_i,
\rho-a_{i-1}\}\le\tilde a_i$ for all $1\le i\le k-1$.
Moreover, since $\alpha \rho/2 \ge (1-B^2F)\alpha r/2 = (1-B^2F)B \delta^{1/2} r^{4/3}/2$, we have
$$
\Pr(\tilde a_i >\rho/2) = e^{-\alpha \rho/2} \le e^{-(1-B^2F)B \delta^{1/2} r^{4/3}/2}\;.
$$

Thus, with probability at least $1-ke^{-(1-B^2F)B \delta^{1/2} r^{4/3}/2}$,  for every $1\le i\le k-1$ we have
$\tilde a_i\le \rho/2$, and therefore $a_i=\tilde a_i$. This event implies that
\begin{equation}\label{ap:abs:boundb}
x_{k-1} = (k-1)\rho - \tilde a,
\end{equation}
where $\tilde a=\sum_{i=1}^{k-1} \tilde a_i$, and also that we did not add any extra vertices (i.e.\
all $v_1,\ldots,v_{k-1}$ belong to the Poisson point process in $[0,\infty)\times[0,\alpha]$).

By construction, each $x_i$ belongs to the rectangle $R'_i\subseteq R_{v_{i-1}}$ for every $1\le
i\le k-1$.
Hence, by Lemma~\ref{ap:lem:rectangle}, the vertices $v_0,v_1,\ldots,v_{k-1}$ form a connected path.

In view of all that, it suffices to show that $x_{k-1}+\rho \ge t$ with sufficiently large probability.
Note that, if this event holds, then $v$ must belong to $R_{v_j}$ for some $0 \leq j
\leq k-1$,  and therefore $u=v_0,v_1,v_2\ldots,v_j,v$ is a connected path of length $j+1\le k$. (Observe that such
a path is contained in $R$, so our extension of the Poisson point
process to $R_\infty$ turned out to be harmless.)

Recall that $C=1/J^{3/2}$. Using the upper-tail bound in Lemma~\ref{ap:lem:concentration} we obtain
$$
\Pr\left(\tilde a\ge (1+C\delta^{3/2})\frac{k}{2\alpha}\right) <
\left(\frac{1+C\delta^{3/2}}{e^{C\delta^{3/2}}}\right)^k  \leq e^{-kg\left((\delta/J)^{3/2}\right)}.
$$
Combining this together with~\eqref{ap:abs:boundb}, we infer that, with probability at least
\[
1-k e^{-(1-B^2F)B \delta^{1/2} r^{4/3}/2} - e^{-kg\left((\delta/J)^{3/2}\right)},
\]
we have
\begin{equation}\label{ap:abs:boundc}
x_{k-1}+\rho =
k\rho - \tilde a
> k\rho - \frac{\left(1+C\delta^{3/2}\right)k}{2\alpha} =
kr\Big(1 - \frac{\alpha^2}{r^2} - \frac{(1+C\delta^{3/2})}{2\alpha r}\Big).
\end{equation}
{From} the definition of $k$, the range of $\delta$ and since $\alpha=B\delta^{1/2}r^{1/3}$, the event above implies
\begin{align*}
x_{k-1}+\rho &>  t(1+\delta r^{-4/3})   \bigg(1 - \delta r^{-4/3}\Big(B^2+\frac{(\delta^{-3/2}+C)}{2B}\Big) \bigg)
\\
&\geq   t(1+\delta r^{-4/3})   \bigg(1 - \delta r^{-4/3}\Big(B^2+\frac{C}{B}\Big) \bigg)
\\
&=  t \Bigg[1 + \delta r^{-4/3}   \bigg(1  - \big(\delta r^{-4/3}+1\big)\Big(B^2+\frac{C}{B}\Big) \bigg) \Bigg]
\\
&\ge  t \Bigg[1 + \delta r^{-4/3}   \bigg(1  - (F+1)\Big(B^2+\frac{C}{B}\Big) \bigg) \Bigg]
>  t,
\end{align*}
so
\begin{align}
\Pr(d_G(u,v) > k) &\le k e^{-(1-B^2F)B \delta^{1/2} r^{4/3}/2} + e^{-kg\left((\delta/J)^{3/2}\right)}
\notag\\
&\le n e^{-(1-B^2F)B \delta^{1/2} r^{4/3}/2} + e^{-g\left((\delta/J)^{3/2}\right)t/r}
\label{ap:eq:xav1}\\
&\le n e^{-(1-B^2(F+1))B \delta^{1/2} r^{4/3}/2} + e^{-g\left((\delta/J)^{3/2}\right)t/r}
\notag\\
&\le n e^{-C(F+1) \delta^{1/2} r^{4/3}/2} + e^{-g\left((\delta/J)^{3/2}\right)t/r},
\label{ap:eq:xav2}
\end{align}
as desired. On the last step we used the fact that $(1 - B^2(F+1))B \ge C(F+1)$, which easily follows from~\eqref{ap:eq:ineB}. This completes the proof of the proposition.
Note that~\eqref{ap:eq:xav1} may be stronger than~\eqref{ap:eq:xav2} if we choose a constant $B$ which satisfies~\eqref{ap:eq:ineB} and maximises $(1-B^2F)B$.
\end{proof}
\begin{proposition}\label{ap:prop:new_upperbound}
Let $\gamma$ as in the statement of Theorem~\ref{ap:thm:main}, and let $\RGR$ be a random geometric graph on $R$, with $u=(0,0)$ and $v=(t,0)$. Suppose that $r\ge 70\sqrt{\log n}$.
Then we have
\[
d_G(u,v) > \left\lceil \frac{t}{r}\left(1+70\gamma (r t)^{-2/3} \right)\right\rceil,
\]
with probability at most $o(n^{-5/2})$.
\end{proposition}
\begin{proof}

First, observe that, if $t\le r$, then $d_G(u,v)=1$, and the statement holds trivially. Thus, we assume henceforth that $t>r$.

Set $B=47/50$, $C=10^{-2}$, $F=23/200$, $D=70$, $E=31$ and $J=10^{-4/3}$.
Recall that
$$
\gamma=\max\left\{E\left(\frac{2r\log{n}}{r+t}\right)^{2/3},D\frac{\log^2{n} }{r^{8/3}}, 3^{2/3}J\right\}\;.
$$
We want to apply Proposition~\ref{ap:alternative} with $\delta=\gamma$. It is straightforward to check that the restrictions~\eqref{ap:eq:ineB} and $J>\frac{3(F+1)}{2^{2/3}}$, required in Proposition~\ref{ap:alternative} hold. We also need to show that $J\le \gamma \leq F r^{4/3}$. Notice that $\frac{D\log^{2}{n}}{r^{8/3}} \leq Fr^{4/3}$, since $r \ge 70\sqrt{\log{n}} \ge (D/F)^{1/4}\sqrt{\log{n}}$; also $E \left(\frac{2r\log{n}}{r+t}\right)^{2/3}\le  Fr^{4/3}$, since $r(r+t)/(2\log n) >r^2/\log n \ge 4900 \ge (E/F)^{3/2}$; and finally $ 3^{2/3}J \leq F r^{4/3}$ since $r=\Omega(\sqrt{n})$. Moreover, $\delta \geq 3^{2/3}J \geq J$.

Note that this choice of constants combined with~\eqref{ap:eq:bound_alpha} and~\eqref{ap:eq:bound_rho} implies
\begin{equation}\label{ap:eq:bound_alpha_rho}
\alpha \le r/3
\qquad\text{and}\qquad
\rho \ge 8r/9 \ge 8\alpha/3.
\end{equation}
The proof concludes by applying~\eqref{ap:eq:xav1} in the proof of Proposition~\ref{ap:alternative} with this given $\delta$, showing that the upper bound on $\Pr(d_G(u,v)>k)$ is $o(n^{-5/2})$.
On the one hand, $\delta\geq \frac{D\log^{2}{n}}{r^{8/3}}$ implies
\[
\frac{(1-B^2F)B\delta^{1/2}r^{4/3}}{2}-\log n \ge \frac{(1-B^2F)BD^{1/2}\log n}{2}-\log n >  \frac{7.01}{2}\log{n} - \log{n} = \frac{5.01}{2}\log{n}.
\]
On the other hand, $\delta\geq E(r\log{n}/t)^{2/3}$  and  $\delta \geq 3^{2/3}J$ imply


\[
\frac{g\left((\delta/J)^{3/2}\right)t}{r} >
\frac{(\delta/J)^{3/2}t}{2r} \ge
\frac{3}{2}CE^{3/2} \log n >
\frac{5.17}{2}\log{n},
\]
where we have used that $g(x)\geq x/2$ if $x\geq 3$.

Therefore, $\Pr(d_G(u,v)>k) \le n^{-5.01/2} + n^{-5.17/2} = o(n^{-5/2})$.
\end{proof}
\begin{corollary}\label{ap:cor:statement2}
 Statement (ii) in Theorem~\ref{ap:thm:main} is true.
\end{corollary}
 \begin{proof}
 Observe that from the proof of Proposition~\ref{ap:alternative} together with~\eqref{ap:eq:bound_alpha_rho}, $x_1\ge \rho/2 > 4\alpha/3$ with
probability at least $1-o(n^{-5/2})$. In particular, this event implies that $v_1$ is outside of the
square $[0,1.01\alpha] \times [0,\alpha]$.
Moreover, also with probability $1-o(n^{-5/2})$, we can find some point $\hat v_j$ in
$[t-1.01\alpha-r/2,t-1.01\alpha]\times[0,\alpha]$.
It may happen that $v_j$ lies in $[t-1.01\alpha,t]\times[0,\alpha]$. However, in that case, we can
replace $v_j$ with $\hat v_j$, and therefore we found a $u$--$v$ path of length $j+1\le k$ with all
internal vertices in $[1.01\alpha,t-1.01\alpha]\times[0,\alpha]$.
Indeed, we will show now that we can always fit such a rectangle
$R'=[1.01\alpha,t-1.01\alpha]\times[0,\alpha]$, suitably rotated and translated, into the square. We
need first a few definitions.

Consider two points $u=(x_u,y_u)$ and $v=(x_v,y_v)$ in $\mathbb{R}^2$. By symmetry we may assume that $x_u < x_v$ and  $y_u
\leq y_v$. Let $\beta$ be the angle of the vector $\vec{uv}$ with respect to the horizontal
axis. Again by symmetry, we may consider $\beta\in [0,\pi/4]$.

We consider now two rectangles of dimensions $\alpha \times t$ placed on each side of the segment
$uv$. Let $R^+$ be the rectangle to the left of $\vec{uv}$, and let $R^-$ be the rectangle to the right of
$\vec{uv}$. We will show that at least one of these rectangles contains a copy of $R'$ fully contained in $\SR$.

Notice that the intersection of $R^+$ and $R^-$ with each of the halfplanes $x \leq
x_u$, $x \geq x_v$, $y \leq y_u$ and $y \geq y_v$ gives $4$ triangles. We call them $T_u^+$, $T_v^-$, $T_u^- $ and $T_v^+$ respectively. All
these triangles are right-angled, and denote by $t_u^+$, $t_v^-$, $t_u^-$ and $t_v^+$ the side of the corresponding triangle that it is
parallel to the segment $uv$. Notice that  $|t_u^+|=|t_v^-|$ and $|t_u^-|=|t_v^+|$. Call a triangle $T_w^*$, with $w \in \{u,v\}$ and $* \in \{+,-\}$,
\emph{safe} if $|t_w^*| \leq 1.01\alpha$.  Note that if $T_u^+$ and $T_v^+$ are safe or fully
contained in the square, then $R^+$ contains the desired rectangle $R$, and analogously for $R^-$.

Since we assumed that $\beta \leq \pi/4$, we have  $|t_u^+| = |t_v^-|= \alpha |\tan \beta| \leq 1.01 \alpha$. Thus, $T_u^+$ and $T_v^-$ are safe.
If $y_u=y_v$, that is $\beta=0$, it is clear that either $R^+$ or $R^-$ contain the desired copy of $R'$. Thus, we may assume that $\beta>0$.

We can also assume that both $u$ and $v$ are on the boundary of $\SR$, as otherwise we extend the
line segment $uv$ to the boundary of the square, and the original rectangles are contained in the
new ones.

Recall that $T_u^+$ and $T_v^-$ are safe.
If $y_v\leq \sqrt{n}/2 -\alpha$, then $T_v^+$ is completely contained in the
square, and hence $R^+$ satisfies the conditions. Similarly, if $y_u\geq -\sqrt{n}/2+\alpha$, $R^-$ satisfies the conditions.
Otherwise, $|y_u|,|y_v|\geq \sqrt{n}/2 -\alpha$, and the angle $\beta$ is at least $\arctan \left(\frac{\sqrt{n}-2\alpha}{\sqrt{n}}\right)>\pi/4$, which contradicts our assumption on $\beta$.

Again, by de-Poissonizing $\RGuv$, we can use Proposition~\ref{ap:prop:new_upperbound} to show that for given $u$ and $v$ in $G\in\RG$,
statement (ii) in Theorem~\ref{ap:thm:main} holds with probability at least $1-o(n^{-2})$. By taking a union bound over all at most $n^2$ possible pairs of vertices, statement (ii) in Theorem~\ref{ap:thm:main} follows.
\end{proof}
\section{Open problems}

Theorem~\ref{ap:thm:main} establishes a relation between the graph distance and the Euclidean distance of two vertices $u$ and $v$ in $\RG$ that holds a.a.s.\ simultaneously for all pairs of vertices.

It would be interesting to find better concentration bounds on the values that $d_G(u,v)$ can take with high probability. Also, we would like to characterize the probability distributions of $\E(d_G(u,v)\mid d_E(u,v))$ and ${\rm Var}(d_G(u,v)\mid d_E(u,v))$ (i.e.~the expectation and variance of $d_G(u,v)$ given $d_E(u,v)$).
What can we say about these distributions?

In the proof of statement~(ii) in Theorem~\ref{ap:thm:main}, we define a new random variable that stochastically dominates $d_G(u,v)$ and we give an upper bound for the probability that this random variable is too large. This argument can be easily adapted in the case $r=\omega(r_c)$, and provide the upper bound
$\E(d_G(u,v)\mid d_E(u,v)) - d_E(u,v)/r= O\left(1+ \gamma d_E(u,v)r^{-7/3} \right)$.
Similarly, the proof of statement~(i) in Theorem~\ref{ap:thm:main} can be adapted to give a lower bound on $\E(d_G(u,v)) \mid d_E(u,v))$, but we need the further conditioning upon the event that $d_E(u,v)$ is large enough.


\begin{thebibliography}{AGE02}
\bibliographystyle{alpha}

\bibitem{Bradonjic}
M.~Bradonjic, R.~Els\"{a}sser, T.~Friedrich, T.~Sauerwald and  A.~Stauffer,
\newblock Efficient broadcast on random geometric graphs.
\newblock {\em Proceedings of the Twenty-First Annual ACM-SIAM Symposium on Discrete Algorithms}, pp. 1412--1421. Society for Industrial and Applied Mathematics, 2010.
%
\bibitem{Ellis}
R.~B.~Ellis, J.~L.~Martin and C.~Yan,
\newblock  Random geometric graph diameter in the unit ball
\newblock {\em Algorithmica}, 47(4), pp. 421--438, 2007.
%
\bibitem{Gilbert}
E.N.~Gilbert,
\newblock Random Plane Networks,
\newblock  {\em  J. Soc. Industrial Applied Mathematics}, 9(5), pp. 533--543, 1961.
%
\bibitem{Goel05}
A.~Goel, S.~Rai and B.~Krishnamachari,
\newblock Sharp thresholds for monotone properties in random geometric graphs,
\newblock  {\em  Annals of Applied Probability}, 15, pp. 364--370, 2005.
 %
\bibitem{GrimmettS}
G.~Grimmett and D.~Stirzaker,
\newblock {\em Probability and Random Processes (3rd Edition)},
 \newblock Oxford U. P., 2001.
%
\bibitem{Friedrich}
T.~Friedrich, T.~Sauerwald and A.~Stauffer,
\newblock Diameter and broadcast time of random geometric graphs in arbitrary dimensions,
\newblock {\em Algorithms and Computation} pp. 190--199,  Springer Berlin Heidelberg, 2011.
%
\bibitem{Abbas}
A.~Merhabian and N.~Wormald,
\newblock On the stretch factor of randomly embedded random graphs,
\newblock  http://arxiv.org/pdf/1205.6252v1.pdf, 2013.
%
\bibitem{Muthukrishnan}
S.~Muthukrishnan and G.~Pandurangan,
\newblock The bin-covering technique for thresholding random geometric graph properties,
\newblock {\em Proceedings of the sixteenth annual ACM-SIAM symposium on Discrete algorithms}, pp. 989--998. Society for Industrial and Applied Mathematics, 2005.
%
\bibitem{Penrose97}
M.~Penrose,
\newblock The longest edge of the random minimal spanning tree,
\newblock  {\em  Annals of Applied Probability}, 7(2), pp. 340--361, 1997.
%
\bibitem{Penrose}
M.~Penrose,
\newblock {\em Random Geometric Graphs},
\newblock Oxford Studies in Probability. Oxford U.P., 2003.
\end{thebibliography}
\end{document}